\documentclass{article}
\usepackage{fullpage}
\usepackage[utf8]{inputenc}
\usepackage{amsmath,amssymb,amsthm,color}
\usepackage{algorithm}
\usepackage{algpseudocode}
\usepackage{enumitem}
\usepackage{tikz}

\newtheorem{theorem}{Theorem}[section]
\newtheorem{definition}{Definition}[theorem]
\newtheorem{lemma}[theorem]{Lemma}

\newcommand{\OPT}{\mathsf{OPT}}
\newcommand{\eps}{\varepsilon}

\newcommand{\convexhull}{\mathcal{CH}}
\newcommand{\ball}{\mathsf{Ball}}
\newcommand{\noise}{\mathsf{Noise}}

\title{Learning Lines with Ordinal Constraints}

\author{
Bohan Fan\thanks{
Department of Computer Science, University of Illinois at Chicago, Chicago, IL, 60607, United States, E-mail: \{bfan4, dihara2, nmoham24, fsgher2, sidiropo, mvaliz2\}@uic.edu. Supported by NSF grants CCF-1815145, CCF-1934915, and by NSF CAREER award 1453472.
}
\and
Diego Ihara Centurion\footnotemark[1]
\and
Neshat Mohammadi\footnotemark[1]
\and
Francesco Sgherzi\footnotemark[1]
\and
Anastasios Sidiropoulos\footnotemark[1]
\and
Mina Valizadeh\footnotemark[1]
}

\begin{document}

\maketitle

\begin{abstract}
We study the problem of finding a mapping $f$ from a set of points into the real line, under ordinal triple constraints.
An ordinal constraint for a triple of points $(u,v,w)$ asserts that $|f(u)-f(v)|<|f(u)-f(w)|$.
We present an approximation algorithm for the dense case of this problem.
Given an instance that admits a solution that satisfies $(1-\eps)$-fraction of all constraints, our algorithm computes a solution that satisfies $(1-O(\eps^{1/8}))$-fraction of all constraints, in time $O(n^7) + (1/\eps)^{O(1/\eps^{1/8})} n$.
\end{abstract}

\section{Introduction}
Geometric methods provide several tools for the analysis of complicated data sets, such as nearest-neighbor search, clustering, and dimensionality reduction.
The key abstraction is to encode a set of objects by mapping each object to a point in some metric space, such that the distance between points quantifies the pairwise  dissimilarity between the corresponding objects.
The success of this paradigm crucially depends on the metrical representation used to encode the data.
Motivated by this fact, metric learning aims at developing methods for discovering an underlying metric space from proximity information (we refer the reader to \cite{shakhnarovich2005learning,kulis2013metric} for a detailed exposition).

There are several different formulations of the metric learning problem that have been considered in the literature.
Here, we focus on the popular case of \emph{ordinal} constraints.
In this case, the input consists of a set of points $X=[n]$, together with a set ${\cal T}$ of ordered triples $(u,v,w)$ of points, representing the fact that $u$ is \emph{more similar} to $v$ than to $w$.
The goal is to find a mapping $f:X\to Y$, for some host metric space $(Y,\rho)$, such that for all $(u,v,w)\in {\cal T}$, we have
\begin{align}
\rho(f(u),f(v)) < \rho(f(u),f(w)). \label{eq:con}
\end{align}
In general, there might be no mapping $f$ that satisfies all constraints of the form \eqref{eq:con}, so we are interested in the algorithmic problem of computing a mapping that minimizes the fraction of violated constraints.
We focus on the case where the host space is the real line, so the objective can be formulated as computing a mapping $f:[n]\to \mathbb{R}$, where for each $(u,v,w)\in {\cal T}$ we have the constraint
\begin{align}
|f(u) - f(v)| < |f(u) - f(w)|. \label{eq:constraint}
\end{align}
We refer to this problem as Line Learning with Ordinal Constraints (LLOC).


\subsection{Our contribution}
We present an approximation algorithm for learning a line metric space under ordinal constraints, for the case of dense instances.
Here, the density condition means that all ordinal information is given, i.e.~for any distinct $u,v,w\in [n]$, we have either $(u,v,w)\in {\cal T}$, or $(u,w,v)\in {\cal T}$.
Our main result is summarized in the following.

\begin{theorem}\label{thm:main}
There exists an algorithm that given an instance of LLOC that admits a solution satisfying $(1-\eps)$-fraction of all constraints, outputs a solution that satisfies $(1-O(\eps^{1/8}))$-fraction of all constraints, in time $O(n^7) + (1/\eps)^{O(1/\eps^{1/8})} n$.
\end{theorem}

\paragraph{Brief overview of our approach.}
The main idea used to obtain Theorem \ref{thm:main} is to first compute an ordering that is close to the ordering of the points in the optimal solution.
This is done by ``guessing'' a point $p^*$ that lies within the few left-most points in an optimal solution, and such that $p^*$ is not involved in many violated constraints. 
We show that the ordinal constraints involving $p^*$ can be used to order the points by first solving an instance of the Minimum Feedback Arc Set problem on a tournament, and then computing a topological ordering of the remaining acyclic graph.
We use this ordering to partition the points into ``buckets'', and
we show that for almost all buckets, almost all their points must be mapped inside an interval that does not contain many other points.
This property allows us to define a smaller instance of the problem by contracting each bucket into a single point.
This new smaller instance can be solved exactly, and its solution can be pulled back to the original problem.

\subsection{Related work}
\paragraph{Metric learning.}
Another popular formulation of the metric learning problem uses \emph{contrastive} constraints.
If this case, the input consists of a set of points $X=[n]$, together with sets ${\cal S}, {\cal D} \subseteq \binom{X}{2}$, where ${\cal S}$ contains pairs labeled as \emph{similar}, and ${\cal D}$ contains pairs labeled as \emph{dissimilar}.
The goal is to find a mapping $f:X\to Y$, for some host metric space $(Y,\rho)$, such that for all $\{u,v\}\in {\cal S}$,
\[
\rho(f(u),f(v)) \leq u,
\]
and for all $\{u,v\}\in {\cal D}$,
\[
\rho(f(u),f(v)) \geq \ell,
\]
for some given threshold values $u,\ell>0$.
This problem is easily seen to be a generalization of Correlation Clustering.
It has been for the case dense instances, when the host metric space is either Euclidean or a tree \cite{ihara2019algorithms}.
The main result of \cite{ihara2019algorithms} is a FPTAS for the case where there exists a mapping that satisfies all constraints, that is allowed to violate the constraints by a small multiplicative factor which is referred to as \emph{contrastive distortion}.
In contrast, in the present work, we do not introduce any distortion, and we do not need to assume that there exists a mapping satisfying all the constraints.

We also note that the case of arbitrary instances (i.e., not necessarily dense) under contrastive constraints has been studied for the setting of learning Mahalanobis metric spaces (i.e., when $X$ is a set of points in $d$-dimensional Euclidean space, and $f$ is required to be linear) \cite{maha}.
This version of the problem is related to the theory of LP-type problems.

\paragraph{Embedding into the line.}
The problem of computing a geometric representation of a data set into the real line has been studied extensively in various forms.
This is arguably the simplest instance of dimensionality reduction, which is also a prototypical unsupervised metric learning task.
Various objectives have been studied, including multiplicative \cite{nayyeri2015reality,DBLP:conf/soda/NayyeriR17,badoiu2005low,badoiu2005approximation,carpenter2018algorithms,fellows2009distortion}, additive \cite{badoiu2003approximation}, and average \cite{DBLP:journals/mst/DhamdhereGR06,rabinovich2003average} distortion.
We refer the read to \cite{DBLP:reference/cg/IndykMS17} for a detailed exposition.
A related notion is ordinal embeddings, where one seeks to obtain mappings that approximately preserve the relative ordering of pairwise distances \cite{alon2008ordinal,buadoiu2008ordinal}.
We remark that a key difference between these works and our result is that they seek to minimize the \emph{ordinal distortion}, which is a multiplicative factor of violation of the ordinal constraints, while we are interested in minimizing the number of violated ordinal constraints (without introducing ordinal distortion).

\paragraph{Betweenness.}
In the Betweenness problem we are given some set $X=[n]$ and a set ${\cal T}$ of ordered triples $(a,b,c)\in [n]^3$.
The goal is to find a bijection $g:[n]\to [n]$ such that for any $(a,b,c)\in {\cal T}$, $g(b)$ appears between $g(a)$ and $g(c)$.
This problem has been studied extensively in the literature.
It is known to be 
MAXSNP-hard \cite{chor1998geometric} (see also \cite{opatrny1979total}),
and remains hard to approximate even on dense instances \cite{ailon2007hardness}.
The case of tournaments has been shown to admit a PTAS \cite{karpinski2011approximation}, while the best approximation algorithm for general instances is the $1/3$-approximation obtained by taking a uniformly random ordering, assuming the Unique Games conjecture \cite{charikar2009every} (see also \cite{makarychev2012simple}).

The Betweenness problem is conceptually similar to the Line Learning with Ordinal Constraints problem studied here.
However, as we now explain, the two problems have some important differences.
A first difference is that the ordinal constraint \eqref{eq:constraint} does not imply any ordering constraint\footnote{For example,  the constraint $(u,v,w)$ is  satisfied by both solutions $f(u)=1,f(v)=2,f(w)=3$, and $f(u)=2,f(v)=1,f(w)=4$, however the former solution implies the ordering $f(u)<f(v)<f(w)$, while the latter implies $f(v)<f(w)<f(w)$}.
A second difference is that the solution space to the Line Learning with Ordinal Constraints problem that we study is larger.
In other words, the ordering of the points is not always enough to recover a nearly-optimal constraint.
For example, consider the instance on $X=\{0,2,4,\ldots,2k,2k+1,\ldots,3k\}$, with all constraints $(u,v,w)\in X^3$, such that $|u-v|<|u-w|$.
Clearly, setting $f$ to be the identity results in a solution that satisfies all constraints.
However, just the ordering of the points in $f$ is not enough to obtain a good solution: setting $g(u_i)=i$, where $g(u_1)<g(u_2)<\ldots<g(u_n)$ results in a solution $g$ that violates a constant fraction of all constraints.

\subsection{Organization}
The rest of the paper is organized as follows.
Section \ref{sec:warmup} presents, as a warm up, an exact polynomial-time algorithm for the case where there exists a solution that satisfies all constraints.
Section \ref{sec:alg} presents the algorithm for the general case.
Section \ref{sec:analysis} presents the analysis.
Section \ref{sec:brittle} gives the proof of a technical Lemma which is used in the proof of the main result.



\section{Warm up: An exact algorithm with no violations}\label{sec:warmup}

We now describe an exact polynomial-time algorithm for the case where there exists an optimal solution that satisfies all constraints.
This algorithm is significantly simpler than the one used to prove our main result.
However, it illustrates the main idea of using the constraints involving some point $p$ to deduce an ordering of all points, and then using this ordering to obtain an embedding into the line.
The algorithm is summarized in the following.

\begin{theorem}
There exists a polynomial-time algorithm which given an instance $([n], {\cal T})$ of the LLOC problem, either computes a mapping $f:[n]\to \mathbb{R}$ that satisfies all the constraints, or correctly decides that no such mapping exists.
\end{theorem}

\begin{proof}
Fix some optimal mapping $f^*:[n]\to \mathbb{R}$, that satisfies all constraints in ${\cal T}$. We guess $p=\arg\min\limits_{x \in [n]}f^*(x)$.
For all $i,j\in [n]$, let $d_{i,j}=|f^*(x_j)-f^*(x_i)|$.
We first determine the ordering of all the points on the real line, and then we compute the mapping using their distance constraints and solving some LP.

Suppose that $[n]=\{x_1,\ldots,x_n\}$, such that
\[
f^*(p)=f^*(x_1) < f^*(x_2) < \ldots < f^*(x_n).
\]
Since $\eps^*=0$, it follows that for all $i<j\in [n]$, we have
$d_{1,i}<d_{1,j}$, and $(1,i,j)\in {\cal T}$.
Therefore, for any $q,q'\in [n]$, we can decide whether $f^*(q)<f^*(q')$ or $f^*(q')<f^*(q)$ based on whether $(p,q,q')\in {\cal T}$ or $(p,q',q)\in {\cal T}$.
Therefore, we can compute the ordering  $x_1,\ldots,x_n$ of $[n]$ by running a sorting algorithm using pairwise comparisons.


We now compute a mapping using an LP.
For any $i<j\in \{1,\ldots,n\}$, we have
$|f^*(x_i)-f^*(x_j)| = \sum_{t=i}^{j-1}d_{t, t+1}$.
Therefore for each $(x_i,x_j,x_k)\in {\cal T}$, the constraint
$|f^*(x_i)-f^*(x_j)| < |f^*(x_i)-f^*(x_k)|$ can be written as
$\sum_{t=i}^{j-1}d_{t, t+1} < \sum_{t=i}^{k-1}d_{t, t+1}$.
Thus computing the desired mapping $f$ can be done by computing a feasible solution to the following LP:
\begin{align*}
  d_{i,j} \geq 0 & \text{ for all } i<j\in [n]\\
  \sum_{t=i}^{j-1}d_{t, t+1} < \sum_{t=i}^{k-1}d_{t, t+1} & \text{ for all } (x_i,x_j,x_k)\in {\cal T}
\end{align*}
This concludes the proof.
\end{proof}

\section{The algorithm for the general case}\label{sec:alg}

In this Section we present the algorithm for the general case of the problem.
The algorithm uses as a subroutine an exact algorithm for a generalized weighted version of the problem.
This exact algorithm is used on small instances that are constructed via a process which we refer to as a \emph{retraction}.

\subsection{Retractions}

We now define a weighted version of the metric learning problem, where each constraint is associated with some weight, and the goal is to maximize the total weight of all satisfied constraints.
Formally, an input to the Weighted Line Learning with Ordinal Constraints (WLLOC) problem is defined by a tuple $([b], {\cal T}, w)$, where $b\in \mathbb{N}$, and ${\cal T}$ are as before, and 
$w:{\cal T} \to \mathbb{R}$ is a weight function.
The goal is to find a solution $f:[b]\to [0,1]$ that minimizes the total weight of violated constraints.

\begin{theorem}\label{thm:weighted_exact}
There exists an exact algorithm for the  WLLOC problem with running time $O(n^{3n})$.
\end{theorem}

\begin{proof}
We identify the space of possible solutions with $[0,1]^n$, by mapping each  solution $f:[b]\to [0,1]$  to the vector $x_f = (f(1),\ldots,f(n)) \in [0,1]^n$.
For any $(i,j,k)\in {\cal T}$, we have the constraint
\[
|f(i)-f(j)| < |f(i)-f(k)|.
\]
The feasible region for this constraint is thus defined as a union of certain cells in an arrangement $A_{(i,j,k)}$ of a constant number of open halfspaces in $\mathbb{R}^n$.
Let $A$ be the arrangement obtained as the union of all halfspaces for all $(i,j,k)\in {\cal T}$.
It is known that any arrangement of $a$ halfspaces in $\mathbb{R}^b$ has complexity $O(a^b)$ (see \cite{toth2017handbook} and references therein), and thus $A$ has complexity $O(|{\cal T}|^n) = O(n^{3n})$.
By enumerating all the cells in this arrangement, we find a  solution that satisfies a set of constraints of maximum total weight, which results in an algorithm with running time $O(n^{3n})$.
\end{proof}

As mentioned earlier, the exact algorithm from Theorem \ref{thm:weighted_exact} will be used as a subroutine on smaller instances.
The following Definition describes a process for mapping large unweighted instances to smaller weighted ones.

\begin{definition}[Retraction]
Given an instance $\phi=([n], {\cal T})$ of the LLOC problem, and some partition  ${\cal B}=\{B_1,\ldots,B_b\}$ of $[n]$, we define the ${\cal B}$-retraction of $\phi$ to be the instance $\phi'=([b], {\cal T}', w)$
of the WLLOC problem
where
for any $(i,j,k)\in {\cal T}'$, we have 
\[
w((i,j,k)) = \left|{\cal T} \cap (B_i\times B_k\times B_j)\right|.
\]
\end{definition}

\subsection{The algorithm}
The last ingredient we need is an approximation algorithm for the Minimum Feedback Arc Set problem on tournaments, which is summarized in the following.

\begin{theorem}[Kenyon-Mathieu \& Schudy \cite{kenyon2007rank}]\label{thm:FAS}
There exists a randomized algorithm for the Minimum Feedback Arc Set problem on weighted tournaments. Given $\epsilon >0$, it outputs a solution with expected cost at most $(1+\epsilon)\OPT$. The expected running time is $O(1/\eps) n^6 + 2^{\tilde{O}(1/\eps)} n^2 + 2^{2^{\tilde{O}(1/\eps)}} n$.
\end{theorem}

We are now ready to describe the general algorithm.
Let ${\cal T}_n$ denote the set of all ordered triples of distinct elements in $[b]$.
Recall that the input consists of a set ${\cal T}\subseteq {\cal T}_n$, such that for any set of distinct $i,j,k\in [b]$, we have that exactly one of the triples $(i,j,k)$ and $(i,k,j)$ is contained in ${\cal T}$.

The algorithm proceeds in the following steps:

\begin{description}
\item{\textbf{Step 1: Exhaustively computing a left-most point.}}
Iterate Steps 2--5 for all values $p\in [b]$.

\item{\textbf{Step 2: Cycle removal.}}
Construct a tournament $G^{(p)}=([b], A^{(p)})$, where
\[
A^{(p)} = \{(i,j) : (p, i, j) \in {\cal T}\}.
\]
Compute an $O(1)$-approximate minimum feedback arc set, $F^{(p)}\subset A^{(p)}$, in $G^{(p)}$, using the algorithm in Theorem \ref{thm:FAS}.

\item{\textbf{Step 3: Ordering.}}
Compute a topological ordering $z_{1}^{(p)},\ldots,z_{n}^{(p)}$ of $G^{(p)}\setminus F^{(p)}$.

\item{\textbf{Step 4: Retraction.}}
Let $b=O(\eps^{-1/8}))$.
For any $i\in [b]$,
let
\begin{align*}
{\cal B}_i^{(p)} &= \bigcup_{j=(i-1)n/b+1}^{in/b} \{z_j^{(p)}\}.
\end{align*}
Let $\psi^{(p)}$ be the ${\cal B}^{(p)}$-retraction of $\phi^{(p)}$.

\item{\textbf{Step 5: Extension.}}
Using the algorithm from Theorem \ref{thm:weighted_exact}, we compute an optimal solution $g:[b]\to [0,1]$ for the instance $\psi^{(p)}$ of WLLOC.
We define $f^{(p)}:[b]\to [0,1]$ by setting for any $i\in [b]$, $f^{(p)}(i)=g^{(p)}(j)$, where $j\in [b]$ such that $i\in B_j^{(p)}$.
The algorithm outputs the solution $f^{(p)}$.

\item{\textbf{Step 6:}}
Return the best solution found among $f^{(1)},\ldots,f^{(n)}$.
\end{description}

This completes the description of the algorithm.

\section{Analysis of the algorithm}
\label{sec:analysis}

This Section presents the analysis of the algorithm, which is the proof of Theorem \ref{thm:main}.

For the remainder of the analysis, let us fix some optimal solution $f_\OPT:[b]\to [0,1]$ for the instance $([b], {\cal T})$ of the LLOC problem.
Fix a numbering $\{x_1,\ldots,x_n\}=[b]$, such that
\[
f_\OPT(x_1) \leq f_\OPT(x_2) \leq \ldots \leq f_\OPT(x_n).
\]

For any $f:[b]\to [0,1]$, 
for any $i\in [b]$, and for any $\alpha\in [0,1]$, 
we say that $i$ is \emph{$\alpha$-good in $f$}, if at least $\alpha$-fraction of the constraints of the form $(i,j,k)\in {\cal T}$ are satisfied; i.e.:
\[
\left|\left\{(i, j, k)\in {\cal T} : |f(i)-f(j)|<|f(i)-f(k)| \right\} \right| \geq \alpha \binom{n-1}{2}.
\]

We first argue that there exists a $(1-\eps^{1/2})$-good point that is close to the left-most point in the optimal solution:

\begin{lemma}\label{lem:good_exists}
There  exists $i^*\in [2\eps^{1/2} n]$, such that $x_{i^*}$ is $(1-\eps^{1/2})$-good in $f_\OPT$.
\end{lemma}

\begin{proof}
Let $\xi$ be the total number of constraints violated by $f_\OPT$.
We have 
$\xi \leq \eps \cdot |{\cal T}| = \eps n \binom{n-1}{2}$.
Suppose that there exists no $i\in [2 \eps^{1/2} n]$ such that $x_i$ is $(1-\eps^{1/2})$-good.
Therefore every $i\in [2\eps^{1/2} n]$ participates in at least $\eps^{1/2} \binom{n-1}{2}$ violated constraints of the form $(i,j,k)$, for some $j,k\in [b]$.
Thus the total number of violated constraints is at least
$\xi \geq 2 n \eps \binom{n-1}{2}$, which is a contradiction, concluding the proof.
\end{proof}

For the remainder of this section, fix some $i^*\in [2\eps^{1/2} n]$, such that $x_{i^*}$ is $(1 - \eps^{1/2})$-good, as in Lemma \ref{lem:good_exists}.
Let $f'$ be the embedding obtained from $f_\OPT$ by exchanging the images of $x_1$ and $x_{i^*}$, that is for all $i\in [b]$,
\[
f'(x_i) = \left\{\begin{array}{ll}
f_\OPT(x_{i^*}) & \text{ if } i=1\\
f_\OPT(x_{1}) & \text{ if } i=i^*\\
f_\OPT(x_i) & \text{ otherwise}
\end{array}\right.
\]
We next show that $f'$ is near-optimal.

\begin{lemma}\label{lem:f_prime_violations}
The total number of violated constraints in $f'$ is at most $(\eps+O(1/n)) n \binom{n-1}{2}$.
\end{lemma}

\begin{proof}
Let ${\cal T}_1\subseteq {\cal T}$ be the set of constraints that are violated in $f'$ and in $f_\OPT$.
Let ${\cal T}_2\subseteq {\cal T}$ be the set of constraints that are violated in $f'$ but not in $f_\OPT$.
We have $|{\cal T}_1| \leq \eps n \binom{n-1}{2}$.
Since $f_\OPT$ and $f'$ differ only on $x_1$ and $x_{i^*}$, it follows that
every constraint $(i,j,k)\in {\cal T}_2$ must contain at least one of $1$ and $i^*$.
There are at most $6 n^2$ such constraints.
Thus $|{\cal T}_2|\leq 6 n^2$.
We conclude that the total number of constraints violated in $f'$ is at most $|{\cal T}_1|+|{\cal T}_2| \leq (\eps+O(1/n)) n \binom{n-1}{2}$, which concludes the proof.
\end{proof}

The next Lemma shows that $x_{i^*}$ remains $(1-O(\eps^{1/2}))$-good in $f'$.

\begin{lemma}\label{lem:pretty_good}
We have that $x_{i^*}$ is $(1-O(\eps^{1/2}))$-good in $f'$.
\end{lemma}

\begin{proof}
Let $\gamma=(x_{i^*},j,k)\in {\cal T}$, and suppose that $\gamma$ is satisfied in $f_\OPT$.
If 
\[
f_\OPT(x_{i^*}) \leq f_\OPT(j) \leq f_\OPT(k),
\]
then, since $f'(j)=f_\OPT(j)$, and $f'(k)=f_\OPT(k)$,
it follows that
\[
f'(x_{i^*}) \leq f'(j) \leq f'(k),
\]
and thus $\gamma$ is also satisfied in $f'$.

Thus, the only possible constraints of the form $(x_{i^*},j,k)\in {\cal T}$, that are not violated in $f_\OPT$, but are violated in $f'$, must satisfy either
$f_\OPT(j)<f_\OPT(x_{i^*})$, 
or 
$f_\OPT(k)<f_\OPT(x_{i^*})$.
In other words, we must have $\{j,k\}\cap \{x_1,\ldots,x_{i^*-1}\} \neq \emptyset$.
Therefore, there are at most $(2\eps^{1/2} n)^2$ such constraints.
Since $x_{i^*}$ is $(1-\eps^{1/2})$-good in $f_\OPT$, it follows that $x_{i^*}$ is $(1-O(\eps^{1/2}))$-good in $f'$, which concludes the proof.
\end{proof}

Let
\[
F' = \{(j,k)\in A^{(i^*)} : (x_{i^*},j,k) \in {\cal T} \text{ and } f' \text{ violates } (x_{i^*},j,k)\}.
\]
The next Lemma shows $F'$ is a valid feedback arc set for $G^{(i^*)}$.

\begin{lemma}\label{lem:F_prime_small}
$F'$ is a feedback arc set for  $G^{(i^*)}$, with $|F'|\leq (O(\eps^{1/2})) \binom{n-1}{2}$.
\end{lemma}

\begin{proof}
By Lemma \ref{lem:pretty_good}, $x_{i^*}$ is $(1-O(\eps^{1/2}))$-good, and thus $|F'| \leq (O(\eps^{1/2})) \binom{n-1}{2}$.
Thus, it suffices to show that $F'$ is a feedback vertex set.
For any $(j,k)\in A^{(i^*)}\setminus F'$, we have that $(x_{i^*}, j, k)$ is satisfied in $f'$.
Since $x_{i^*}$ is mapped to the left-most point in $f'$, it follows that 
$f'(j)<f'(k)$.
It follows that
\[
x_{i^*}, x_2,x_3,\ldots,x_{i^*-1},x_1,x_{i^*+1}, x_{i^*+2},\ldots,x_n
\]
is a topological ordering of $G^{(i^*)}\setminus F'$, and thus $F'$ is a feedback arc set, which concludes the proof.
\end{proof}

If the instance admits a solution with no violations, then it can be shown that the bucketing ${\cal B}^{(i^*)}$ computed by the algorithm agrees with a partition of the optimal solution to contiguous disjoint intervals.
In the following, we show that, in the general case, the bucketing is ``close'' to such a partition.
First, we introduce a notion of ``stability'' which formalizes what it means for a bucket to be close to an optimal interval.

\begin{definition}[Stability]\label{def:stability}
Let $i\in [b]$.
We say that $i$ is \emph{stable} if there exists some interval $I\subset \mathbb{R}$, such that
\[
\left|I \cap f'\left(B_i^{(i^*)}\right)\right| \geq (1-\eps^{1/8}) \cdot n/b,
\]
and
\[
\left|I \cap f'\left([b]\setminus B_i^{(i^*)}\right)\right| \leq \eps^{1/8} \cdot n/b,
\]
We also say that $i$ is $I$-stable.
We say that $i$ is \emph{unstable} ($I$-unstable) if it is not stable ($I$-stable).
\end{definition}

The following Lemma gives a characterization of unstable buckets.

\begin{lemma}\label{lem:unstable_intervals}
Suppose that $i\in [b]$ is unstable.
Then there exist pairwise disjoint intervals $I_1, I_2, I_3 \subset \mathbb{R}$, that appear in this order from left to right in the line, 
such that
\[
\left|I_1 \cap f'\left(B_i^{(i^*)}\right)\right| \geq n \eps^{1/8}/(2b),
\]
\[
\left|I_3 \cap f'\left(B_i^{(i^*)}\right)\right| \geq n \eps^{1/8}/(2b),
\]
and
\[
\left|I_2 \cap f'\left([b]\setminus B_i^{(i^*)}\right)\right| > \eps^{1/8} \cdot n/b,
\]
\end{lemma}

\begin{proof}
Let $I_1\subset \mathbb{R}$ be the minimal interval that contains the $n \eps^{1/8}/(2b)$ left-most points in $f'(B_i^{(i^*)})$,
and let 
$I_3\subset \mathbb{R}$ be the minimal interval that contains the $n \eps^{1/8}/(2b)$ right-most points in $f'(B_i^{(i^*)})$.
Let $I_2\subset \mathbb{R}$ be the maximal interval that is contained between $I_1$ and $I_3$.
Since $\eps^{1/8}< 1$, we have that $I_1\cap I_3=\emptyset$, and therefore, all intervals $I_1$, $I_2$, $I_3$ are well-defined and pairwise disjoint.
By construction, $I_1$ and $I_3$ each contains exactly $n \eps^{1/8} / (2N)$ points in $f'(B_i^{(i^*)})$.
Therefore, it remains to show that $I_2$ contains more than $\eps^{1/8} n/b$ points in $f'([b]\setminus B_i^{(i^*)})$.
Suppose, for the sake of contradiction, that $I_2$ contains at most $\eps^{1/8} n/b$ in $f'([b]\setminus B_i^{(i^*)})$.
Then, $I_2$ contains exactly $(1-\eps^{1/8})n/b$ points in in $f'(B_i^{(i^*)})$, 
and at most $\eps^{1/8}n/b$ points in $f'([b]\setminus B_i^{(i^*)})$,
implying that $B_i$ is stable, which is a contradiction.
This concludes the proof.
\end{proof}

We next show that for each unstable bucket, the feedback arc set must contain many edges incident to vertices in the bucket.

\begin{lemma}\label{lem:max_cut_edges}
Let $i\in [b]$ be unstable.
Then, $F^{(i^*)}\cup F'$ contains at least $\eps^{1/4} n^2/(2b^2)$ arcs having exactly one endpoint in $B_i^{(i^*)}$.
\end{lemma}

\begin{proof}
Let $I_1,I_2,I_3\subset \mathbb{R}$ be the intervals given by Lemma \ref{lem:unstable_intervals}.
Let $v\in [b] \setminus B_i^{(i^*)}$, such that $f'(v)\in I_2$.
Pick $j\in [b]$, such that $v\in B_j^{(i^*)}$.
We consider two cases:

\emph{Case 1: Suppose that $j<i$.}
Let $u\in B_i^{(i^*)}$, such that $f'(u)\in I_1$.
If $(v,u)\in A^{(i^*)}$, then it follows that $f'$ violates $(x_{i^*}, v, u)$, and thus $(v,u)\in F'$.
Otherwise, we have $(u,v)\in A^{(i^*)}$. Since $u$ appears after $v$ in the topological sort of $G^{(i^*)}\setminus F^{(i^*)}$, it follows that $(u,v)\in F^{(i^*)}$.
Thus, in either case, $F^{(i^*)}\cup F'$ contains either $(u,v)$ or $(v,u)$.
Therefore, $F^{(i^*)}\cup F'$ contains at least $n\eps^{1/8}/(2b)$ arcs having $u$ as an endpoint.

\emph{Case 2: Suppose that $j>i$.}
This case is similar to Case 1, and is included for completeness.
Let $u\in B_i^{(i^*)}$, such that $f'(u)\in I_3$.
If $(u,v)\in A^{(i^*)}$, then it follows that $f'$ violates $(x_{i^*}, u,v)$, and thus $(u,v)\in F'$.
Otherwise, we have $(v,u)\in A^{(i^*)}$. Since $u$ appears before $v$ in the topological sort of $G^{(i^*)}\setminus F^{(i^*)}$, it follows that $(v,u)\in F^{(i^*)}$.
Thus, in either case, $F^{(i^*)}\cup F'$ contains either $(u,v)$ or $(v,u)$.
Therefore, $F^{(i^*)}\cup F'$ contains at least $n\eps^{1/8}/(2b)$ arcs having $u$ as an endpoint.

We conclude that, in either case, for any $u\in B_i^{(i^*)}$, $F^{(i^*)}\cup F'$ contains at least $n\eps^{1/8}/(2b)$ arcs having $u$ as an endpoint.
Summing over all $u\in B_i^{(i^*)}$, we obtain that $F^{(i^*)}\cup F'$ contains at least $\eps^{1/4}n^2/(2b^2)$ arcs having an endpoint in $B_i^{(i^*)}$.
This concludes the proof.
\end{proof}

Next, we bound the number of unstable buckets.

\begin{lemma}\label{lem:unstable_N} Let $J = \{i\in [b]:i \text{ is unstable}\}$,
we have  $|J|\leq O(\eps^{1/4}) 2b^2$.
\end{lemma}

\begin{proof}
By Lemma \ref{lem:pretty_good}
we have that $x_{i^*}$ is $(1 - O(\eps^{1/2}))$-good in $f'$, and
by Lemma \ref{lem:F_prime_small} we have that $G^{(i^*)}$ admits a feedback arc set of size at most $(O(\eps^{1/2}))\binom{n-1}{2}$.
Thus, by Theorem \ref{thm:FAS}, the algorithm computes some feedback arc set $F^{(i^*)}\subset A^{(i^*)}$, with $|F^{(i^*)}| = O(\eps^{1/2} n^2)$.
We note that here we only use Theorem \ref{thm:FAS} to obtain a $O(1)$-approximation.
By Lemma \ref{lem:max_cut_edges}, 
\begin{align*}
|J| &\leq |F^{(i^*)}\cup F'| / (\eps^{1/4}n^2/(2b^2))\\
 &\leq O(\eps^{1/4}) 2b^2,
\end{align*}
which concludes the proof.
\end{proof}

For any stable $i\in [b]$, let ${\cal I}_i\subset \mathbb{R}$ be the interval that contains at least $(1-\eps^{1/8}) n/b$ points in $f'(B_i^{(i^*)})$, and at most $\eps^{1/8} n/b$ other points.
Let also ${\cal J}_i\subset {\cal I}_i$ be an open interval that contains all but the $\eps^{1/8}n/b$ leftmost points in $f'(B_i^{(i^*)})\cap {\cal I}_i$, and the $\eps^{1/8}n/b$ rightmost points in $f'(B_i^{(i^*)})\cap {\cal I}_i$.
Thus, $|{\cal J}_i \cap f'(B_i^{(i^*)})| \geq (1-3\eps^{1/8})n/b$.
It follows that for any $i\neq j\in [b]$, such that both $i$ and $j$ are stable, we have
${\cal J}_i\cap {\cal J}_j=\emptyset$.

Intuitively, we intend to find a solution that satisfies a nearly-optimal fraction of constraints, while ignoring all constraints that involve points that are mapped outside the intervals ${\cal J}_i$, where $i\in [b]$ is stable.
To that end, we define a small set of points that the analysis can safely ``ignore'':
\[
X_\noise = \bigcup_{i\in [b] : i \text{ stable}} \left\{v\in B_{i}^{(i^*)} : f'(v)\notin {\cal J}_i \right\}.
\]
Since $|{\cal J}_i \cap f'(B_i^{(i^*)})| \geq (1-3\eps^{1/8})n/b$, it follows that
\begin{align}
|X_\noise| &\leq 3\eps^{1/8} n  \label{eq:X_noise}
\end{align}
Let also, for any $i\in [b]$, 
\[
\bar{B}_i^{(i^*)} = B_i^{(i^*)} \setminus X_\noise.
\]

We identify a set of triples $(i,j,k)\in [b]^3$ for which, intuitively, it is difficult to satisfy at least some significant fraction of all constraints with one point from each of the clusters $B_i^{(i^*)}$, $B_j^{(i^*)}$, and $B_k^{(i^*)}$.
Formally, we say that some $(i,j,k)\in [b]^3$ is \emph{brittle} if there exist 
$u,u'\in {\cal J}_i$, $v,v'\in {\cal J}_j$, and $w,w'\in {\cal J}_k$, such that 
\[
|u-v|<|u-w|,
\]
and
\[
|u'-v'|>|u'-w'|.
\]
Intuitively, the above property implies that if for all $t\in [b]$, all points in $\bar{B}_t^{(i^*)}$ get mapped to the same point  $p_t\in {\cal J}_t$, then there exist choices for the points $\{p_t\}_t$, such that some constraint in $\bar{B}_i^{(i^*)} \times \bar{B}_j^{(i^*)} \times \bar{B}_k^{(i^*)}$ is violated; in other words, if a triple $(i,j,k)$ is not brittle, then the choice of the points $p_t$ does not affect the satisfiability of the constraints in $\bar{B}_i^{(i^*)} \times \bar{B}_j^{(i^*)} \times \bar{B}_k^{(i^*)}$.

We are now ready to show that the retraction computed by the algorithm admits a solution of low total cost.

\begin{lemma}\label{lem:retraction_cost}
The instance $\psi^{(i^*)}$ of $WLLOC$ constructed in Step 4 admits a solution that satisfies constraints of total weight at least $|{\cal T}| \cdot (1- O(\eps^{1/8}))$.
\end{lemma}

\begin{proof}
We define a mappings $g:[b]\to [0,1]$, and $g':[b]\to [0,1]$, as follows.
For each $i\in [b]$, pick $v_i\in {\cal B}_i^{(i^*)}$, arbitrarily, and set 
\[
g'(i)=f'(v_i).
\]
For any $j\in [b]$, we set
\[
g(j) = g'(i),
\]
where $i\in [b]$ is the unique integer such that $i\in B_i^{(i^*)}$.
By the definition of the WLLOC instance $\psi^{(i^*)}$, the total weight of the constraints violated by $g'$ equals the total number of constraints violated by $g$.
It therefore suffices to upper bound the number of constraints in ${\cal T}$ that are violated by $g$.

We define a partition ${\cal T}={\cal T}_0 \cup {\cal T}_1 \cup {\cal T}_2 \cup {\cal T}_3 \cup {\cal T}_4$, where
\begin{align*}
{\cal T}_0 &= \{(u,v,w)\in {\cal T} : f' \text{ violates } (u,v,w)\}, \\
{\cal T}_1 &= \{(u,v,w)\in {\cal T} : \text{ at least two of } u,v,w \text{ are in the same cluster in } {\cal B}^{(i^*)}\}, \\
{\cal T}_2 &= \{(u,v,w)\in {\cal T} : u\in B_i^{(i^*)}, v\in B_j^{(i^*)}, w\in B_k^{(i^*)}, \text{ and at least one of } i,j,k \text{ is unstable }\}, \\
{\cal T}_3 &= \{(u,v,w)\in {\cal T} : u\in B_i^{(i^*)}, v\in B_j^{(i^*)}, w\in B_k^{(i^*)}, \text{ and } (i,j,k)  \text{ is brittle}\}, \\
{\cal T}_4 &= \{(u,v,w)\in {\cal T} : \{u,v,w\}\cap X_{\noise} \neq \emptyset\}, \\ 
{\cal T}_5 &= {\cal T} \setminus ({\cal T}_0 \cup {\cal T}_1 \cup {\cal T}_2 \cup {\cal T}_3 \cup {\cal T}_4). 
\end{align*}
By Lemma \ref{lem:f_prime_violations} we have
\begin{align*}
|{\cal T}_1| &\leq (\eps + O(1/n)) n \binom{n-1}{2}.
\end{align*}
Since every cluster in ${\cal B}^{(i^*)}$ has $n/b$ points, we have
\begin{align}
|{\cal T}_1| &\leq 3 n^3/b^2. \label{eq:T1}
\end{align}
In order to bound $|{\cal T}_3|$ we need a bound on the number of brittle triples.
This is done in Lemma \ref{lem:unstable_N}, which appears in Section \ref{sec:brittle}.
We thus have
\begin{align}
|{\cal T}_2| &\leq O(\eps^{1/4}) 2 n^3 b. \label{eq:T2}
\end{align}
By Lemma \ref{lem:brittle_N} we have 
\begin{align}
|{\cal T}_3| &\leq n^3/b. \label{eq:T3}
\end{align}
By \eqref{eq:X_noise} we have 
\begin{align}
|{\cal T}_4| &\leq O(\eps^{1/8}) n^3 \label{eq:T4}
\end{align}

Let $(u,v,w)\in {\cal T}_5$.
By the definition of ${\cal T}_5$, we have that 
$u\in \bar{B}_i^{(i^*)}$,
$v\in \bar{B}_j^{(i^*)}$,
and
$w\in \bar{B}_k^{(i^*)}$,
for some distinct $i,j,k\in [b]$,
such that $(i,j,k)$ is not brittle,
and $f'$ satisfies $(u,v,w)$, that is
\[
|f'(u)-f'(v)| < |f'(u)-f'(w)|.
\]
By the definition of a brittle tripple we get
\[
|g(u)-g(v)| < |g(u)-g(w)|,
\]
and thus $g$ satisfies $(u,v,w)$.
We obtain that $g$ satisfies all constraints in ${\cal T}_5$.
Thus, by \eqref{eq:T1}--\eqref{eq:T4}, the number of constraints violated by $g$ is at most
$|{\cal T}_0|+\ldots+|{\cal T}_4| \leq n^3 O(\eps^{1/8})$, which concludes the proof.
\end{proof}


We are now ready to prove our main result.

\begin{proof}[Proof of Theorem \ref{thm:main}]
By Lemma \ref{lem:retraction_cost} we have that WLLOC instance $\psi^{(i^*)}=([b],{\cal T}',w)$ constructed at Step 4 of the algorithm, admits a mapping $g':[b]\to \mathbb{R}$, such that the total weight of the constraints in ${\cal T}'$ violated by $g'$ is at most $O(\eps^{1/8} n^3)$.
Therefore, in Step 5, using the exact algorithm from Theorem \ref{thm:weighted_exact}, we compute a mapping $g:[b]\to \mathbb{R}$, violating the same total weight as $g'$.
By the definition of retraction, it follows that the mapping $f^{(i^*)}$ computed in Step 5 violates at most $O(\eps^{1/8} n^3)$ constraints in ${\cal T}$, as required.

It remains to bound the running time.
Step 2 uses the algorithm from Theorem \ref{thm:FAS} to obtain a $O(1)$-approximate minimum feedback arc set, and thus takes time $O(n^6)$.
Step 3 takes time $O(n)$ and Step 4 takes time $O(n^2)$.
Step 5 runs the algorithm from Theorem \ref{thm:weighted_exact} on an input of size $N$, and thus takes time $O(N^{3N}) + O(n)$.
Step 6 requires computing the number of violated constraints in each of the $n$ solutions, and thus takes total time $O(n^4)$.
Due to Step 1, the Steps 2--5 are repeated $n$ times, and thus the total running time is at most $O(n^7 + b^{3b} n) = O(n^7) + (1/\eps)^{O(1/\eps^{1/8})} n$, which concludes the proof.
\end{proof}

\section{Bounding the number of brittle triples}
\label{sec:brittle}

This Section is devoted to proving an upper bound on the number of brittle triples.
We begin by deriving a simple condition that is a consequence of brittleness.


\begin{lemma}\label{lem:brittle_equidistant}
Let $j<i<k\in [b]$.
We have that if $(i,j,k)$ is brittle, then there exist $p_i\in {\cal J}_i$, $p_j\in {\cal J}_j$, $p_k\in {\cal J}_k$, such that \[
p_i-p_j = p_k-p_i.
\]
\end{lemma}

\begin{proof}
If $(i,j,k)$ is brittle, it is easy to see that ${\cal J}_i$ must be located between ${\cal J}_j$ and ${\cal J}_k$; otherwise, any representative point chosen in ${\cal J}_i$ must be closer to all the points in ${\cal J}_j$ than those in ${\cal J}_k$, or vice versa. By definition, there exist $p_i\in {\cal J}_i$, $p_j'\in {\cal J}_j$, $p_k'\in {\cal J}_k$, such that \[
p_i-p_j' \ge p_k'-p_i,
\] and $p_j''\in {\cal J}_j$, $p_k''\in {\cal J}_k$, such that \[
p_i-p_j'' < p_k''-p_i.
\]
Without loss of generality, assume $p'_j<p''_j$ and $p'_k<p''_k$, and define $\delta\in[0,1]$. Comparing $d_{ij}(\delta)=p_i-(p'_j+\delta(p''_j-p'_j))$ and $d_{ik}(\delta)=(p'_k+\delta(p''_k-p'_k))-p_i$, we have $d_{ij}(0)-d_{ik}(0)\ge 0$ and $d_{ij}(1)-d_{ik}(1)<0$. There exist $\delta'\in[0,1]$, s.t. $d_{ij}(\delta')-d_{ik}(\delta')=0$. 

Define $p_j=(p'_j+\delta'(p''_j-p'_j))\in {\cal J}_i$ and $p_k=(p'_k+\delta'(p''_k-p'_k))\in {\cal J}_j$, we have \[
p_i-p_j = p_k-p_i,
\]
which concludes the proof.
\end{proof}

\begin{lemma}\label{lem:hyperrectangles_intersection}
Let $i_1,i_2,i_3,j_1,j_2,j_3,k_1,k_2,k_3\in \mathbb{R}$, with $i_1<i_2<i_3$, $j_1<j_2<j_3$, $k_1<k_2<k_3$.
For any $\alpha,\beta,\gamma\in \{1,2\}$, let $H_{\alpha,\beta,\gamma}$ be the axis-parallel parallelepiped defined by
\[
H_{\alpha,\beta,\gamma} := \convexhull(\{(i_{\alpha+\alpha'}, j_{\beta+\beta'}, k_{\gamma+\gamma'}) : \alpha',\beta',\gamma' \in \{0,1\}\}).
\]
Let $h$ be any plane in $\mathbb{R}^3$.
Then, there exist $\alpha^*,\beta^*,\gamma^* \in \{0,1\}$, such that
$h$ does not intersect the interior of $H_{\alpha^*,\beta^*,\gamma^*}$.
\end{lemma}

\begin{proof}
For any $d\geq 2$, any $d$-dimensional halfspace containing the origin must also contain at least one $d$-orthant.
The assertion follows immediately from the case $d=3$.
\end{proof}

\begin{lemma}\label{lem:brittle_cube}
Let $i,j,k\in [b]$, with $j+1<i$, and $i+1<k$.
Then, there exist $i',j',k'\in \{0,1\}$ such that $(i+i', j+j', k+k')$ is not brittle.
\end{lemma}

\begin{proof}
Define the plane
\[
h = \{(x_I,x_J,x_K)\in \mathbb{R}^3 : x_I-x_J = x_K-x_I\}.
\]
By Lemma \ref{lem:brittle_equidistant}, we have that if $(i+i', j+j', k+k')$ is brittle, then $h$ must intersect the hyperrectangle ${\cal J}_{i+i'}\times {\cal J}_{j+j'} \times {\cal J}_{k+k'}$.
However, by Lemma \ref{lem:hyperrectangles_intersection}, it follows that there exist $i',j',k'\in \{0,1\}$, such that $h$ does not intersect ${\cal J}_{i+i'}\times {\cal J}_{j+j'} \times {\cal J}_{k+k'}$, and thus $(i+i', j+j', k+k')$ is not brittle, which concludes the proof.
\end{proof}

\begin{lemma}[Brittle convexity]\label{lem:brittle_convexity}
Let $\{e_1,e_2,e_3\}$ be the standard orthonormal basis in $\mathbb{R}^3$.
Let $v\in [b-2]^3$,
and let $w\in \{e_1,e_2,e_3\}$, 
such that $v$ and $v+2w$ are both brittle.
Then, $v+w$ is also brittle.
\end{lemma}

\begin{proof}
By Lemma \ref{lem:brittle_equidistant}, there exist $p_i\in {\cal J}_i$, $p_j\in {\cal J}_j$, $p_k\in {\cal J}_k$, such that 
\begin{align}
p_i-p_j = p_k-p_i. \label{eq:equi_p}
\end{align}
Let $w=(i',j',k')$.
Similarly, there exist
$q_i\in {\cal J}_{i+2i'}$, $q_j\in {\cal J}_{j+2j'}$, $q_k\in {\cal J}_{k+2k'}$, such that 
\begin{align}
q_i-q_j = q_k-q_i. \label{eq:equi_q} 
\end{align}
For any $\alpha\in [0,1]$, let
\[
z^{(\alpha)}_i = (1-\alpha)p_i + \alpha q_i,
\]
\[
z^{(\alpha)}_j = (1-\alpha)p_j + \alpha q_j,
\]
\[
z^{(\alpha)}_k = (1-\alpha)p_k + \alpha q_k.
\]

Let us assume that $w=e_1$.
The cases $w=e_2$ and $w=e_3$ can be handled in a similar manner.
We have that for all $\alpha\in [0,1]$, $z_j^{(\alpha)} \in {\cal J}_j$, and $z_k^{(\alpha)} \in {\cal J}_k$.
Moreover, $z_i^{(0)}\in {\cal J}_i$, and $z_i^{(1)}\in {\cal J}_{i+2}$, which implies that there exists some $\alpha^*\in [0,1]$, such that
$z_{i}^{(\alpha^*)}\in {\cal J}_{i+1}$.
We have
\begin{align*}
z_i^{(\alpha^*)} - z_j^{(\alpha^*)} &= (1-\alpha^*)p_i + \alpha^* q_i - (1-\alpha^*)p_j - \alpha^* q_j \\
 &= (1-\alpha^*)(p_i-p_j) + \alpha^* (q_i-q_j) \\
 &= (1-\alpha^*)(p_k-p_i) + \alpha^* (q_k-q_i) \\
 &= (1-\alpha^*)p_k + \alpha^* q_k - (1-\alpha^*)p_i - \alpha^* q_i \\
 &= z_k^{(\alpha^*)} - z_i^{(\alpha^*)},
\end{align*}
which by Lemma \ref{lem:brittle_equidistant} implies that $v+w$ is brittle, and concludes the proof.
\end{proof}

We are now ready to bound the number of brittle triples, which is the main result of this Section.

\begin{lemma}\label{lem:brittle_N}
The number of brittle triples is at most $O(b^2)$.
\end{lemma}

\begin{proof}
Let 
$B\subseteq [b]^3$
be the set of all brittle 
triples, and let $B'=[b]^3\setminus B$.
For any $s\in \{0,1\}^3$, let
\[
U_s = s\cdot  b/2 + [b/2]^3,
\]
and $B_s = B\cap U_s$.
Since $B=\bigcup_{s} B_s$, and there are only 8 different values for $s$, it suffices to show that for any $s\in \{0,1\}^3$, $|B_s| = O(b^2)$.
We shall prove this for the case $s=(0,0,0)$. All remaining cases can be handled in a similar manner.

For the remainder for the proof, let $s=(0,0,0)$.
By Lemma \ref{lem:brittle_cube}, it follows that for any $v\in B_s^3$, there exists $v'\in B'$, with $v'-v\in \{0,1\}^3$.
This implies that there exists $u\in B$, and $u'\in B'$, with $u-v\in \{0,1\}^3$, $u'-v\in \{0,1\}^3$, and $u'-u\in \{e_1,e_2,e_3\}$, where $\{e_1,e_2,e_3\}$ is the standard orthonormal basis in $\mathbb{R}^3$.
Let $t=u'-u$.
By Lemma \ref{lem:brittle_convexity}, it follows by induction that for any $i\in \{1,\ldots,b/2\}$, the triple $u+i\cdot t$ is brittle.
Let
\[
R_v = \bigcup_{i=1}^{b/2} \{u+c\cdot i\}.
\]
Thus $R_v\subseteq B'$.
Note that, since $s=(0,0,0)$, we have
\begin{align}
|R_v| \geq b/2. \label{eq:Rv_size}
\end{align}
For any $j\in \{1,2,3\}$, we say that $v$ is \emph{type-$j$}, if $t=e_j$.

Let
\[
B_{s,j} = \{v\in B_s : v \text{ is type-}j\}.
\]
Let $j^*\in \{1,2,3\}$, such that
$|B_{s,j^*}| \geq |B_s|/3$.

By the above construction, it follows that for any $v,w\in B_{s,j^*}$, with $\|v-w\|_{\infty}\geq 2$, we have
$R_v\cap R_w = \emptyset$.
We greedily construct some $C\subseteq B_{s,j^*}$ as follows.
We start with $C:=\emptyset$, and $D:=B_{s,j^*}$.
While $D\neq \emptyset$, we pick any $v\in D$, and we set $C:=C\cup \{v\}$, and $D:=D\setminus \ball_{\infty}(v,1)$, where $\ball_{\infty}(v,r)$ denotes the $\ell_\infty$-ball of radius $r$ centered at $v$.
For every $v$ added to $C$, we delete at most $9$ elements from $D$, and thus 
\[
|C|\geq |B_{s,j^*}|/9 \geq |B_s|/27.
\]
Since for any $v,w\in C$, we have $\|v-w\|_\infty$, it follows that $R_v\cap R_w=\emptyset$.
Combining with \eqref{eq:Rv_size}, we get
\[
b^3 \geq |B'| \geq \left| \bigcup_{v\in C} R_v \right| = \sum_{v\in C} |R_v| \geq |C|\cdot b/2 \geq |B_s| \cdot b / 54,
\]
and thus $|B_s|\leq 54 b^2$, which concludes the proof.
\end{proof}

\bibliographystyle{alpha}
\bibliography{bibfile}

\end{document}